\documentclass{article}
\usepackage[hidelinks]{hyperref}
\usepackage{graphicx} 
\usepackage[T1]{fontenc}
\usepackage{textcomp}
\usepackage{microtype}
\usepackage[margin=1.5in]{geometry}
\usepackage{amsthm}
\usepackage{times}
\usepackage{newtxmath}
\usepackage[natbib=true,style=authoryear,maxcitenames=2,maxbibnames=6,]{biblatex}
\newcommand{\M}{\mathcal M}
\newcommand{\bernoulli}{\mathrm{Bernoulli}}
\newcommand{\binomial}{\mathrm{Binomial}}
\newcommand{\Var}{\mathrm{Var}}

\newtheorem{theorem}{Theorem}

\title{Easier Estimation of Extremes under Randomized Response}
\author{Jonathan Hehir}
\date{\today}

\begin{document}

\maketitle

\begin{abstract}
    \noindent In this brief note, we consider estimation of the bitwise combination $x_1 \lor \dots \lor x_n = \max_i x_i$ observing a set of noisy bits $\tilde x_i \in \{0, 1\}$ that represent the true, unobserved bits $x_i \in \{0, 1\}$ under randomized response.
    We demonstrate that various existing estimators for the extreme bit, including those based on computationally costly estimates of the sum of bits, can be reduced to a simple closed form computed in linear time (in $n$) and constant space, including in an online fashion as new $\tilde x_i$ are observed.
    In particular, we derive such an estimator and provide its variance using only elementary techniques.
\end{abstract}

\section{Introduction}

Consider the simple task of calculating the size of the union of sets from a (small) finite universe, given binary representations of the sets. Suppose we have $n$ sets $S_j \subseteq [m] = \{1, \dots, m\}$ for $j \in [n]$, and we encode variables $x_{ij} = 1_{ i \in S_j }$ for $i \in [m], j \in [n]$. The cardinality of the set union is given by
\[
|S_1 \cup \dots \cup S_n| = \sum_{i=1}^m 1_{i \in S_1 \cup \dots \cup S_n} = \sum_{i=1}^m x_{i1} \lor \dots \lor x_{in} ,
\]
where $x_{i1} \lor \dots \lor x_{in} = \max_j x_{ij}$.

Now suppose that instead of observing the true sets $ \{S_j\}$ or their bit representations $\{x_{ij}\}$, we observe the noisy bit representations $\{\M_{q_{ij}}(x_{ij})\}$, where $\M_{q} : \{0,1\} \to \{0, 1\}$ denotes the \textit{randomized response} mechanism
\[
\M_{q}(x) \overset{ind}{\sim} \begin{cases}
    \bernoulli(1-q), & x = 1 \\
    \bernoulli(q), & x = 0
\end{cases} .
\]
Observing only the noisy bits, we can no longer determine $|S_1 \cup \dots \cup S_n|$ exactly; however, an unbiased estimator for $x_{i1} \lor \dots \lor x_{in}$ leads naturally to an unbiased estimator for $|S_1 \cup \dots \cup S_n|$. Moreover, the variance of this cardinality estimator is equal to the sum of the variances of the $m$ bitwise estimates, and the cardinality estimator is a minimum-variance unbiased estimator (MVUE) if the bitwise estimator is also an MVUE.

This basic problem lies at the heart of several more sophisticated problems. For example, privately estimating the size of set unions from large universes may be achieved efficiently with binary data sketches such as Bloom filters and Flajolet--Martin sketches perturbed by randomized response \citep[e.g.,][]{stanojevic2017distributed,alaggan2018privacy,kreuter2020privacy,hehir2023sketch}, and the resulting estimators generally rely on precisely the sort of estimation outlined above.

In the remainder of this note, we broadly consider estimation of $x_1 \lor \dots \lor x_n = \max_i x_i$ observing only the noisy bits $\{ \M_{q_i}(x_i) \}$. While we focus on the maximum case (i.e., bitwise \textit{or}) in this note, our results and discussion generalize easily to the other extreme, $x_1 \land \dots \land x_n = \min x_i$ (i.e., bitwise \textit{and}). We consider three different estimators arising from existing literature and show that despite their differences, all of these methods may be used to produce identical estimates for $x_1 \lor \dots \lor x_n$. We rederive an equivalent estimator (Theorems~\ref{thm:estimator}, \ref{thm:equivalence}) using elementary techniques to obtain an estimator that (a) can be calculated in $O(n)$ time and $O(1)$ space, (b) can be updated in an online fashion as new noisy bits are observed, (c) flexibly accommodates different randomized response parameters for each bit, and (d) has tractable variance whose form we provide (Theorem~\ref{thm:variance}).

Our estimator may be directly plugged into previous works to obtain more flexible, efficient, streamable estimates. For example, in differentially private cardinality estimation, substituting our estimator in \citet{stanojevic2017distributed} allows the solution to scale beyond two sketches, in \citet{kreuter2020privacy} allows for sketches with different privacy parameters to be combined, and in \citet{hehir2023sketch} leads to a variance reduction when using multiple sketches in exchange for at most a constant space tradeoff.

\section{Existing Estimators for $x_1 \lor \dots \lor x_n$}

\subsection{Method 1: A Sum Estimator}

The problem of estimating $x_1 \lor \dots \lor x_n$ given $\{ \M_{q_i}(x_i) \}$ is perhaps most frequently studied in the simplified case when $q_1 = \dots = q_n = q$ for some known $q \in(0, 1/2)$, in which case we may slightly abuse notation to write:
\begin{equation}
\label{eq:convolution}
\sum_{i = 1}^n \M_q(x_i) \sim \textstyle
    \binomial\left( \sum_i x_i, 1-q \right) + \binomial\left( n - \sum_i x_i, q \right) .
\end{equation}
Thus, in the equal-$q$ case, the sum of the noisy bits is a convolution of two binomial-distributed random variables whose count parameters (in particular $\sum_i x_i$) might become the target of estimation.

A popular method for estimating $\sum_i x_i$ in this case is to derive the probability mass function for the convolution in Eq. (\ref{eq:convolution}) and encode the relevant values into an $(n+1) \times (n+1)$ transition probability matrix $P$, where each entry denotes the probability that a given sum of $x_i$ bits would translate to a sum of noisy $\M_q(x_i)$ bits, i.e., the matrix that satisfies
\[
P e_{1 + \sum_i x_i} = E \left[ e_{1 + \sum_i \M_q(x_i)} \right] ,
\]
where $e_j$ denotes the $j$-th elementary basis vector, i.e., the vector whose $j$-th entry is 1 and remaining entries 0. From here, a method-of-moments estimator yields:
\[
e_{1 + \sum_i x_i} \approx P^{-1} e_{1 + \sum \M_q(x_i)} .
\]
Since $x_1 \lor \dots \lor x_n = 1_{\sum_i x_i > 0}$, this estimator for $\sum_i x_i$ is easily adapted to provide an estimate for $x_1 \lor \dots \lor x_n$:
\[
\hat Y_{convolution} = 1 - P^{-1}_{1,1 + \sum \M_q(x_i)} .
\]
This is the estimator used in \citet{kreuter2020privacy} and discussed in \citet{alaggan2017non}.%
\footnote{\citet{alaggan2017non}, which is the basis for \citet{alaggan2018privacy}, derives an estimator based on $P^{-1}$ that constrains estimates to a suitable space. These works propose solutions for more general counting problems involving noisy bits, of which the type we focus on here is an important special case.} It is easy to see that $\hat Y_{convolution}$ is an unbiased estimator of $x_1 \lor \dots \lor x_n$, as:
\begin{align*}
P e_{1 + \sum_i x_i} &= E \left[ e_{1 + \sum_i \M_q(x_i)} \right] \\
\implies e_{1 + \sum_i x_i} &= E \left[ P^{-1} e_{1 + \sum_i \M_q(x_i)} \right] \\
\implies 1_{\sum_i x_i = 0} &= E \left[ P^{-1}_{1, 1 + \sum_i \M_q(x_i)} \right] \\
\implies 1 - x_1 \lor \dots \lor x_n &= E \left[ 1 - \hat Y_{convolution} \right] .
\end{align*}

We pause to note several areas for improvement in this estimator. First, it is not particularly efficient: obtaining an estimate over $n$ bits requires deriving the convolution described above, then constructing and inverting an $(n+1) \times (n+1)$ matrix---from which only a single entry is used. Second, since a closed form for $P^{-1}$ is non-obvious, the variance of this estimator does not appear to be described in closed form. Finally, this estimator requires that all the $q_i$ are equal and that all bits are considered simultaneously in an offline fashion.

\subsection{Method 2: A Bigger Sum Estimator}

An estimator for $\sum x_i$ is given in \citet{vinterbo2018simple} that generalizes to the case when $q_1, \dots, q_n$ are not all equal. This estimator relies instead on a $2^n \times 2^n$ transition probability matrix that encodes the probability of any sequence of true bits $x_1, \dots, x_n$ being mapped to a given sequence of noisy bits $\M_{q_1}(x_1), \dots, \M_{q_n}(x_n)$. As a result, the presented estimator is even more computationally expensive, requiring exponential time and space. Although not the focus of the original work, this estimator can be adapted to generalize $\hat Y_{convolution}$ to the unequal-$q$ case.

The work of \citet{vinterbo2018simple} essentially generalizes an estimation procedure designed in \citet{stanojevic2017distributed} for the estimation of the cardinality of two sets' union or intersection under local differential privacy.%
\footnote{This same procedure is used in \citet{gao2020dplcf}.}
\citet{stanojevic2017distributed} stop short of generalizing their results to $n \geq 2$ sets, noting the exponential complexity of the solution.

As in the previous method, from the entire inverted transition probability matrix, only a single entry is relevant to the estimation of $x_1 \lor \dots \lor x_n$. Rather than explicitly derive the estimator arising from this method, we will move on to our next method, which obtains that single matrix entry more efficiently.

\subsection{Method 3: Bitwise Operations under Randomized Response}

\citet{hehir2023sketch} demonstrate efficient methods to perform bitwise operations under randomized response, including $\lor$ via a randomized merging algorithm. The aim of their work is not to directly estimate $x_1 \lor \dots \lor x_n$, but rather to produce random bits whose distribution is a function of $x_1 \lor \dots \lor x_n$. In particular, their algorithm $g$ satisfies:
\[
g_{q_1, \dots, q_n}(\M_{q_1}(x_1), \dots, \M_{q_n}(x_n)) \overset D= \M_{q^*}(x_1 \lor \dots \lor x_n),
\]
for some $q^* \in [\max q_i, 1/2)$ given as a function of $q_1, \dots, q_n$. This method is easily adapted to obtain unbiased estimates for $x_1 \lor \dots \lor x_n$ by de-biasing the effects of the randomized response mechanism $\M_{q^*}$, i.e.:
\[
\hat Y_{randmerge} = \frac{g_{q_1, \dots, q_n}(\M_{q_1}(x_1), \dots, \M_{q_n}(x_n)) - q^*}{1 - 2 q^*} .
\]

It's important to note that $\hat Y_{randmerge}$ has two sources of randomness: the randomness from applying $\M_{q_i}$ to the original bits $x_i$, plus additional randomness in $g$ used to perform the aggregation of bits. In the original work, this randomness was imposed to force $g$ to take values in $\{0, 1\}$, but $\hat Y_{randmerge}$ neither requires nor satisfies this property. Accordingly, we adapt the estimator one step further by removing the randomness of the merge operation $g$:
\[
\hat Y_{RBmerge} = \frac{E_g[ g_{q_1, \dots, q_n}(\M_{q_1}(x_1), \dots, \M_{q_n}(x_n))] - q^*}{1 - 2 q^*} .
\]
In the above, the expectation $E_g[ \cdot ]$ is taken with respect to the randomness of the merge operation $g$ only. The fact that $\hat Y_{RBmerge}$ remains unbiased follows from the law of total expectation. (Equivalently, $\hat Y_{RBmerge}$ may be seen as a Rao--Blackwellization of $\hat Y_{randmerge}$.)
\[
E [ \hat Y_{RBmerge} ] = E [ E_g [ \hat Y_{randmerge} ] ] = E [ \hat Y_{randmerge} ] = x_1 \lor \dots \lor x_n .
\]

In fact, it follows from the derivation of $g$ in \citet{hehir2023sketch} that $\hat Y_{RBmerge}$ once again comes from the top row of the inverse of a large transition probability matrix. In particular, the $n$-way merge of \citet[Theorem~A.2]{hehir2023sketch} involves the same $2^n \times 2^n$ transition probability matrix as \citet{vinterbo2018simple}. However, as the authors demonstrate, this merge operation can be performed inductively by combining pairs of bits $n-1$ times, reducing the computational cost from exponential to linear in $n$ and allowing for online/streaming implementations.

While this estimator achieves the flexiblity and computational requirements that we desire, it remains somewhat complicated and lacks an explicit variance. We tackle these issues next.

\section{An Elementary Estimator}

We will now derive an estimator for $x_1 \lor \dots \lor x_n$ using elementary techniques. The basic idea is to build an estimator inductively. When $n=1$, we have only one bit, which (trivially) is the extreme. In this case, the estimator simply de-biases the single noisy bit. On the other hand, when $n > 1$, additional de-biased bits may be combined through a simple multiplicative step. Starting with a single bit, observe that:
\begin{align*}
E[ \M_{q_i}(x_i) ] &= q_i (1 - x_i) + (1 - q_i) x_i \\
    &= q_i + x_i (1 - 2 q_i).
\end{align*}
As a result:
\[
E \left[ \frac{\M_{q_i}(x_i) - q_i}{1 - 2 q_i} \right] = x_i.
\]
So for $n=1$, we have an estimator that is unbiased for $x_1 = \min_i x_i = \max_i x_i$:
\[
\hat Y_{elementary}^{(1)} = \frac{\M_{q_1}(x_1) - q_1}{1 - 2 q_1}.
\]
Next, observe that the extreme bits have convenient multiplicative forms:
\begin{align*}
\min_i x_i &= \prod_{i=1}^n x_i , \\
\max_i x_i &= 1 - \prod_{i=1}^n (1 - x_i) .
\end{align*}
Since the randomness of the noisy bits is independent, we have, for example:
\[
E \left[ \prod_{i=1}^n \frac{\M_{q_i}(x_i) - q_i}{1 - 2 q_i} \right] = \prod_{i=1}^n E \left[ \frac{\M_{q_i}(x_i) - q_i}{1 - 2 q_i} \right] = \prod_{i=1}^n x_i = \min_i x_i .
\]
This is neat, but our intended focus was on estimating the maximum bit. Writing
\[
1 - \frac{\M_{q_i}(x_i) - q_i}{1 - 2 q_i} = \frac{1 - q_i - \M_{q_i}(x_i)}{1 - 2 q_i},
\]
we obtain:
\[
E \left[ \prod_{i=1}^n \frac{1 - q_i - \M_{q_i}(x_i)}{1 - 2 q_i} \right] = \prod_{i=1}^n E \left[ 1 - \frac{\M_{q_i}(x_i) - q_i}{1 - 2 q_i} \right] = \prod_{i=1}^n (1 - x_i) = 1 - \max_i x_i .
\]
This leads to the estimator we desire.
\begin{theorem}
\label{thm:estimator}
Suppose $x_1, \dots x_n \in \{0, 1\}$, $q_1, \dots, q_n \in [0, 1/2)$, and let
\[
\hat Y_{elementary} = 1 - \prod_{i=1}^n \frac{1 - q_i - \M_{q_i}(x_i)}{1 - 2 q_i}.
\]
Then $\hat Y_{elementary}$ is unbiased for $x_1 \lor \dots \lor x_n$, i.e.:
\[
E [\hat Y_{elementary} ] = \max_i x_i = x_1 \lor \dots \lor x_n .
\]
\end{theorem}
\begin{proof}
    See above.
\end{proof}

\subsection{Equivalence and Variance of the Estimators}

A general equivalence of $\hat Y_{elementary}$ with the estimators arising from \citet{vinterbo2018simple} and \citet{hehir2023sketch} can be seen through a cumbersome process of inspection. In short, letting $\otimes$ denote the Kronecker product, those estimators may be found by constructing the matrix
\[
K^{-1} = \left( \bigotimes_{i=1}^n \begin{bmatrix} 1-q_i & q_i \\ q_i & 1-q_i \end{bmatrix} \right)^{-1} = \prod_{i=1}^n(1-2q_i)^{-1} \bigotimes_{i=1}^n \begin{bmatrix} 1-q_i & -q_i \\ -q_i & 1-q_i \end{bmatrix} ,
\]
then choosing a value from its top row according to the specific observed sequence $\M_{q_1}(x_1), \dots, \M_{q_n}(x_n)$ and subtracting that value from $1$ \citep[see, e.g.,][Theorem~A.2]{hehir2023sketch}. The resulting estimator is precisely $\hat Y_{elementary}$.

Having stated the above without formal proof, we provide a more formal result for the case when $q_1 = \dots = q_n$. (Recall that $\hat Y_{convolution}$ is only defined in this case.)

\begin{theorem}
    \label{thm:equivalence}
    Suppose $x_1, \dots x_n \in \{0, 1\}$, $q_1 = \dots = q_n = q \in (0, 1/2)$, and consider $\hat Y_{convolution}$, $\hat Y_{RBmerge}$, and $\hat Y_{elementary}$ as previously described. Then $\hat Y_{convolution} = \hat Y_{RBmerge} = \hat Y_{elementary}$.
\end{theorem}

\begin{proof}
    We avoid deriving the convolution in Eq.~\ref{eq:convolution} \citep[see][Appendix~A]{alaggan2017non} and turn instead to statistical estimation theory to prove these estimators are equivalent.
    
    Although the theorem concerns arbitrary $x_1, \dots x_n \in \{0, 1\}$, suppose for a moment that $X_1, \dots, X_n$ are random variables, independent and identically distributed as $\bernoulli(\theta)$. Let $\tilde X_i = \M_q(X_i)$. Then $\tilde X_i$ are also i.i.d. Bernoulli with parameter:
    \[
    \mu = E[ \tilde X_i ] = E[ E[ \M_q(X_i) \mid X_i ] ] = E[ q + (1 - 2q) X_i] = q + (1 - 2q) \theta .
    \]
    It is well known that $\sum_i \tilde X_i$ is a complete, sufficient statistic for $\mu$ \citep[e.g.,][Example~6.2.22]{casella2002statistical}. Consequently, it is also complete and sufficient for $\theta$. As a result, the Lehmann--Scheff{\'e} theorem states that any function $f(\sum \tilde X_i)$ that is an unbiased estimator for a function $\tau(\theta)$ is the \textit{unique} minimum-variance unbiased estimator for $\tau(\theta)$. Clearly, $\hat Y_{convolution}$ is a function of $\sum_i \tilde X_i$. It can be seen that $\hat Y_{RBmerge}$ is as well, through the inspection process described in the preceding paragraphs. In fact, in the equal-$q$ case, we have:
    \[
    \hat Y_{RBmerge} = \hat Y_{elementary} = 1 - \frac{q^{\sum_i \tilde X_i} (1 - q)^{n - \sum_i \tilde X_i}}{(1 - 2q)^n} .
    \]
    So indeed, all three estimators are a function of $\sum_i \tilde X_i$. We will now show that all these estimators are unbiased estimators for the same quantity, $1 - (1 - \theta)^n$, and therefore are identical estimators.
    
    Recall that for fixed $x_1, \dots, x_n$, each of the estimators in question is unbiased for $x_1 \lor \dots \lor x_n$. In other words for each estimator $\hat Y \in \{ \hat Y_{elementary}, \hat Y_{RBmerge}, \hat Y_{convolution} \}$:
    \begin{align*}
    E[\hat Y] &= E [ E[ \hat Y \mid X_1, \dots, X_n ] ] \\
        &= E[ X_1 \lor \dots \lor X_n ] \\
        &= 1 - P(X_1 = 0, \dots, X_n = 0) \\
        &= 1 - \prod_{i = 1}^n P(X_i = 0) \\
        &= 1 - (1 - \theta)^n .
    \end{align*}
\end{proof}

Given the convenient closed-form of $\hat Y_{elementary}$ and the equivalence of the three estimators, the variance of these estimators may now readily be found, also using elementary techniques.

\begin{theorem}
    \label{thm:variance}
    Suppose $x_1, \dots x_n \in \{0, 1\}$, $q_1, \dots, q_n \in [0, 1/2)$. Then
    \[
    \Var(\hat Y_{elementary}) = \prod_i \left( 1 - x_i + \frac{q_i (1 - q_i)}{(1 - 2 q_i)^2} \right) - 1_{\sum x_i = 0} .
    \]
\end{theorem}

\begin{proof}
    To begin, let $Z_i = 1 - q_i - \M_{q_i}(x_i)$, and note that:
    \begin{align*}
        \Var[Z_i] &= \Var[\M_{q_i}(x_i)] \\
            &= q_i(1-q_i) \\
        E[Z_i] &= 1 - q_i - E[\M_{q_i}(x_i)] \\
            &=\begin{cases}
                0, & x_i = 1 \\
                1 - 2q_i, & x_i = 0
            \end{cases} \\
            &= (1 - x_i) (1 - 2 q_i) \\
        E[Z_i^2] &= \Var(Z_i) + (E[Z_i])^2 \\
            &= q_i (1 - q_i) + (1 - x_i) (1 - 2 q_i)^2 .
    \end{align*}
    Then the variance of $\hat Y_{elementary}$ is given by:
    \begin{align*}
    \Var(\hat Y_{elementary}) &= \Var \left( \prod_{i=1}^n\frac{1 - q_i - \M_{q_i}(x_i)}{1-2q_i} \right) \\
        &= \frac{\Var \left( \prod_i Z_i \right)}{\prod_i (1-2q_i)^2} \\
        &= \frac{\prod_i E \left[ Z_i^2 \right] - \prod_i (E[Z_i])^2}{\prod_i (1-2q_i)^2} \\
        &= \frac{\prod_i \left( q_i (1 - q_i) + (1 - x_i)(1 - 2q_i)^2 \right) - \prod_i (1-x_i)(1-2q_i)^2}{\prod_i (1-2q_i)^2} \\
        &= \prod_i \left( \frac{q_i (1 - q_i)}{(1 - 2 q_i)^2} + (1 - x_i) \right) - \prod_i (1-x_i) \\
        &= \prod_i \left( 1 - x_i + \frac{q_i (1 - q_i)}{(1 - 2 q_i)^2} \right) - 1_{\sum x_i = 0} .
    \end{align*}
\end{proof}

\printbibliography

@article{kreuter2020privacy,
  title={Privacy-preserving secure cardinality and frequency estimation},
  author={Kreuter, Benjamin and Wright, Craig William and Skvortsov, Evgeny Sergeevich and Mirisola, Raimundo and Wang, Yao},
  year={2020}
}

@article{hehir2023sketch,
  title={Sketch-Flip-Merge: Mergeable Sketches for Private Distinct Counting},
  author={Hehir, Jonathan and Ting, Daniel and Cormode, Graham},
  journal={arXiv preprint arXiv:2302.02056},
  year={2023}
}

@inproceedings{alaggan2017non,
  title={Non-interactive (t, n)-incidence counting from differentially private indicator vectors},
  author={Alaggan, Mohammad and Cunche, Mathieu and Minier, Marine},
  booktitle={Proceedings of the 3rd ACM on International Workshop on Security And Privacy Analytics},
  pages={1--9},
  year={2017}
}

@article{alaggan2018privacy,
  title={Privacy-preserving wi-fi analytics},
  author={Alaggan, Mohammad and Cunche, Mathieu and Gambs, S{\'e}bastien},
  journal={Proceedings on Privacy Enhancing Technologies},
  volume={2018},
  number={2},
  pages={4--26},
  year={2018}
}

@inproceedings{vinterbo2018simple,
  title={A simple algorithm for estimating distribution parameters from n-dimensional randomized binary responses},
  author={Vinterbo, Staal A},
  booktitle={Information Security: 21st International Conference, ISC 2018, Guildford, UK, September 9--12, 2018, Proceedings},
  pages={192--209},
  year={2018},
  organization={Springer}
}

@inproceedings{gao2020dplcf,
  title={DPLCF: differentially private local collaborative filtering},
  author={Gao, Chen and Huang, Chao and Lin, Dongsheng and Jin, Depeng and Li, Yong},
  booktitle={Proceedings of the 43rd International ACM SIGIR Conference on Research and Development in Information Retrieval},
  pages={961--970},
  year={2020}
}

@inproceedings{stanojevic2017distributed,
  title={Distributed cardinality estimation of set operations with differential privacy},
  author={Stanojevic, Rade and Nabeel, Mohamed and Yu, Ting},
  booktitle={2017 IEEE Symposium on Privacy-Aware Computing (PAC)},
  pages={37--48},
  year={2017},
  organization={IEEE}
}

@book{casella2002statistical,
  title={Statistical Inference},
  author={Casella, George and Berger, Roger L},
  year={2002},
  publisher={Thomson Learning}
}

\end{document}